\documentclass{ifacconf}  

\usepackage{mathtools} 
\usepackage{graphicx} 
\usepackage{amsmath} 
\usepackage{amssymb}  
\usepackage{natbib}
\usepackage{multirow}
\usepackage{algorithm, algorithmicx, algpseudocode}

\newcommand{\R}{{\mathbb R}}

\newcommand{\N}{{\mathbb N}}
\newcommand{\Rnn}{{\mathbb R}_{\ge 0}}
\newcommand{\Rp}{{\mathbb R}_{> 0}}
\newcommand{\C}{{\mathbb C}}

\newcommand{\cA}{{\mathcal A}}

\newcommand{\cB}{{\mathcal B}}

\newcommand{\cD}{{\mathcal D}}

\newcommand{\cK}{{\mathcal K}}

\newcommand{\cS}{{\mathcal S}}
\newcommand{\cT}{{\mathcal T}}
\newcommand{\cU}{{\mathcal U}}

\newcommand{\diag}{{\mathrm{diag}}}

\newcommand{\cl}{\mathrm{cl}}

\newcommand{\inter}{\mathrm{int}}

\def\QED{\mbox{\rule[0pt]{1.3ex}{1.3ex}}} 

\newenvironment{proof}{{\quad \it Proof:\,}}{\hfill \QED \par}

\newenvironment{proof-of}[1]{{\quad\it Proof of #1:\,}}{\hfill\QED\par}

\begin{document}

\begin{frontmatter}
\title{Shaping Pulses to Control Bistable Monotone Systems Using Koopman Operator}
\author[First]{Aivar Sootla} \and \author[Second]{Alexandre Mauroy} \and \author[Third]{Jorge Goncalves}

\address[First]{Montefiore Institute, University of Li\`{e}ge, B-4000 Li\`{e}ge, Belgium (e-mail: asootla@ulg.ac.be).}
\address[Second]{Luxembourg Centre for Systems Biomedicine, University of Luxembourg, L-4367 Belvaux, Luxembourg (e-mail: alexandre.mauroy@uni.lu)}
\address[Third]{Luxembourg Centre for Systems Biomedicine, University of Luxembourg, L-4367 Belvaux, Luxembourg (e-mail: jorge.goncalves@uni.lu)}

\thanks[footnoteinfo]{A. Sootla holds an F.R.S--FNRS fellowship. This work was performed when A. Mauroy was with the University of Li\`{e}ge and held a return grant from the Belgian Science Policy (BELSPO).}
\begin{abstract}
In this paper, we further develop a recently proposed control method to switch a bistable system between its steady states using temporal pulses. The motivation for using pulses comes from biomedical and biological applications (e.g. synthetic biology), where it is generally difficult to build feedback control systems due to technical limitations in sensing and actuation. The original framework was derived for monotone systems and all the extensions relied on monotone systems theory. In contrast, we introduce the concept of switching function which is related to eigenfunctions of the so-called Koopman operator subject to a fixed control pulse. Using the level sets of the switching function we can 
(i) compute the set of all pulses that drive the system toward the steady state in a synchronous way and (ii) estimate the time needed by the flow to reach an epsilon neighborhood of the target steady state. Additionally, we show that for monotone systems the switching function is also monotone in some sense, a property that can yield efficient algorithms to compute it. This observation recovers and further extends the results of the original framework, which we illustrate on numerical examples inspired by biological applications.
\end{abstract}
\end{frontmatter}

\section{Introduction}
In many applications, the use of a time-varying feedback control signal is impeded by the limitations in sensing and actuation. One of such applications is synthetic biology, which aims to engineer and control biological functions in living cells~(\cite{brophy2014principles}), and which is an emerging field of science with applications in metabolic engineering, bioremediation and energy sector~(\cite{Purnick:2009}). Recently, control theoretic regulation of protein levels in microbes was shown to be possible by~\cite{milias2011silico, Menolascina:2011, uhlendorf2012long}. However, the proposed methods are hard to automate due to physical constraints in sensing and actuation (for example, using the techniques from~\cite{Levskaya09,Mettetal08}). In the context of actuation, adding a chemical solution to the culture is fairly straightforward, but in contrast, removing a chemical from the culture is much more complicated (this could be done through diluting, but would be time consuming and hard to perform repeatedly). In regard to these constraints, it is therefore desirable to derive control policies which can not only solve a problem (perhaps not optimally) but are also simple enough to be implemented in an experimental setting.

The pioneering development in synthetic biology was the design of the so-called \emph{genetic toggle switch} by~\cite{Gardner00}, which is a synthetic genetic system (or a circuit) of two mutually repressive genes \emph{LacI} and \emph{TetR}. Mutual repression means that only one of the genes can be activated or switched ``on'' at a time. The activated gene expresses proteins within a cell, hence the number of proteins expressed by the ``on'' gene is much higher than the number of proteins of the ``off'' gene. This entails the possibility of modeling this genetic circuit by a bistable dynamical system. Since toggle switches serve as one of the major building blocks in synthetic biology, we set up the control problem of switching from one stable fixed point to another (or toggling a gene). Recently, \cite{sootla2015pulsesacc,sootla2015pulsesaut} proposed to solve the problem using temporal pulses $u(t)$ with fixed length $\tau$ and magnitude $\mu$: 
\begin{equation}\label{eq:pulse}
u(t) = \mu h(t,\tau) \qquad 
h(t,\tau) = 
\begin{cases} 1 & 0 \leq t \leq \tau\,,\\
0 & t >\tau\,.
\end{cases}
\end{equation}
In the case of monotone systems (cf.~\cite{angeli2003monotone}), the set of all pairs $(\mu,\tau)$ allowing a switch (i.e. \emph{the switching set}) was completely characterized. In particular, the boundary of this set, called \emph{the switching separatrix}, was shown to be monotone, a result which significantly simplifies the computation of the switching set. However, the contributions of~\cite{sootla2015pulsesaut} provide only a binary answer (on whether a given control pulse switches the system or not), but do not characterize the time needed to converge to the steady state. 

In this paper, we conduct a theoretical study which extends the results by~\cite{sootla2015pulsesaut} and provides a temporal characterization of the effects of switching pulses. To do so, we exploit the framework of the so-called Koopman operator (cf.~\cite{mezic2005}), which is a linear infinite dimensional representation of a nonlinear dynamical system. In particular, we use the spectral properties of the operator, focusing on the Koopman eigenfunctions (i.e. infinite dimensional eigenvectors of the operator). We first introduce the \emph{switching function}, which we define as a function of $\mu$ and $\tau$ related to the dominant Koopman eigenfunction. Each level set of the switching function characterizes a set of pairs $(\mu, \tau)$ describing control pulses that drive the system synchronously to the target fixed point. Hence, the switching function provides a temporal characterization of the controlled trajectories. The switching separatrix introduced by~\cite{sootla2015pulsesaut} is interpreted in this framework as a particular level set of the switching function. Furthermore, there is a direct relationship between the level sets of the switching function and the so-called isostables introduced in~\cite{mauroy2013isostables}.

Since the switching function is defined through a Koopman eigenfunction, it can be computed in the Koopman operator framework with numerical methods based on Laplace averages. These methods can be applied to a very general class of systems, but usually require extensive simulations. However, we show that the switching function of monotone systems is also monotone in some sense, so that its level sets can be computed in a very efficient manner by using the algorithm proposed in~\cite{sootla2016basins}. The key to reducing the computational complexity is to exploit the properties of the Koopman eigenfunctions of a monotone system.

The main contribution of this paper is to provide a theoretical framework that relates the Koopman operator to control problems. We note, however, that the eigenfunctions of the Koopman operator can be estimated directly from the observed data using dynamic mode decomposition methods (cf.~\cite{Schmid2010,Tu2014}). Therefore our results could potentially be extended to a data-based setting, which would increase their applicability.


The rest of the paper is organized as follows. In Section~\ref{s:prel}, we cover some basics of monotone systems theory and Koopman operator theory. In Section~\ref{s:shaping-pulses}, we review the shaping pulses framework from~\cite{sootla2015pulsesaut} and present the main results of this paper. We illustrate the theoretical results on examples in Section~\ref{s:example}.

\section{Preliminaries}
\label{s:prel}
Consider control systems in the following form
\begin{equation}
\label{sys:f}
\dot x = f(x,u),\quad x(0) = x_0,
\end{equation} 
with $f: \cD\times \cU\rightarrow \R^n$, $u:\R_{\ge 0}\rightarrow \cU$, and where $\cD\subset\R^n$, $\cU\subset\R$ and $u$ belongs to the space $\cU_{\infty}$ of Lebesgue measurable functions with values from $\cU$. We define the flow map $\phi: \R \times \cD \times \cU_{\infty}\rightarrow \R^n$, where $\phi(t, x_0, u)$ is a solution to the system~\eqref{sys:f} with an initial condition $x_0$ and a control signal $u$. If $u=0$, then we call the system~\eqref{sys:f} \emph{unforced}. We denote the Jacobian matrix of $f(x,0)$ as $J(x)$. If $x^\ast$ is a fixed point of the unforced system, we assume that the eigenvectors of $J(x^\ast)$ are linearly independent, for the sake of simplicity. We denote the eigenvalues of $J(x^\ast)$ by $\lambda_i$.

\emph{Koopman Operator.} Spectral properties of nonlinear dynamical systems can be described through an operator-theoretic framework that relies on the so-called Koopman operator $L = f^T \nabla$, which is an operator acting on the functions $g:\R^n\rightarrow \C$ (also called observables). We limit our study of the Koopman operator to unforced systems~\eqref{sys:f} (that is, with $u=0$) on a basin of attraction $\cB\subset\R^n$ of a stable hyperbolic fixed point $x^\ast$ (that is, the eigenvalues $\lambda_j$ of the Jacobian matrix $J(x^\ast)$ are such that $\Re(\lambda_j) < 0$ for all $j$). In this case, the Koopman operator admits a point spectrum and the eigenvalues $\lambda_j$ of the Jacobian matrix $J(x^\ast)$ are also eigenvalues of the Koopman operator. In the non-hyperbolic case, the analysis is  more involved since the spectrum of the Koopman operator may be continuous. The operator $L$ generates a semigroup acting on observables $g$
\begin{gather}
U^t g(x) = g \circ \phi(t,x,0),
\end{gather}
where $\circ$ is the composition of functions and $\phi(t,x,0)$ is a solution to the unforced system for $x\in\cB$. Since the operator is linear (cf.~\cite{mezic2013analysis}), it is natural to study its spectral properties. In particular, the eigenfunctions $s_j: \cB \mapsto \C$ of the Koopman operator are defined as the functions satisfying $L s_j = f^T \nabla s_j=\lambda_j s_j$, or equivalently
\begin{gather}
\label{eq:property_eigenf}
U^t s_j(x) = s_j(\phi(t, x,0)) = s_j(x) \, e^{\lambda_j t}, \quad x\in \cB,
\end{gather}
where $\lambda_j\in \mathbb{C}$ is the associated eigenvalue.

If the vector field $f$ is a $C^2$ function, then the eigenfunctions $s_j$ are $C^1$ functions~(\cite{mauroy2014global}). If the vector field $f$ is analytic and if the eigenvalues $\lambda_j$ are simple, the flow of the system can be expressed through the following expansion (see e.g.~\cite{mauroy2013isostables}):
\begin{align}
&\phi(t, x,0) = x^\ast + \sum\limits_{j=1}^n s_j(x) v_j e^{\lambda_j t} + \label{eq:expansion}\\
&\notag\sum\limits_{\begin{smallmatrix}
k_1,\dots,k_n\in\N_0\\
k_1+\dots+k_n>1 
\end{smallmatrix}} v_{k_1,\dots,k_n} \, s_1^{k_1}(x) \cdots s_n^{k_n}(x) e^{(k_1\lambda_1 +\dots k_n\lambda_n) t}, 
\end{align}
where $\N_0$ is the set of nonnegative integers, $\lambda_j$, $v_j$ are the eigenvalues and right eigenvectors of the Jacobian matrix $J(x^\ast)$, respectively, and the vectors $v_{k_1,\dots,k_n}$ are the Koopman modes (see~\cite{mezic2005, mauroy2014global} for more details). Note that a similar expansion can also be obtained if the eigenvalues are not simple (cf.~\cite{mezic2015applications}).

Throughout the paper we assume that $\lambda_j$ are such that $0>\Re(\lambda_1) > \Re(\lambda_j)$ for all $j\ge 2$. In this case, the eigenfunction $s_1$, which we call a dominant eigenfunction, can be computed through the so-called Laplace average
\begin{gather}\label{laplace-average}
g_\lambda^\ast(x) = \lim\limits_{t\rightarrow \infty}\frac{1}{T}\int\limits_0^T (g\circ \phi(t, x,0)) e^{-\lambda t} d t.
\end{gather}
For all $g$ that satisfy $g(x^\ast)=0$ and $\nabla g(x^\ast) \cdot v_1 \neq 0$, the Laplace average $g_{\lambda_1}^\ast$ is equal to $s_1(x)$ up to a multiplication with a scalar. If we let $g(x) = w_1^T (x-x^\ast)$, where $w_1$ is the left eigenvector of $J(x^\ast)$ corresponding to $\lambda_1$, the limit in~\eqref{laplace-average} does not converge if $x\not\in\cB$. Therefore, we do not require the knowledge of $\cB$ in order to compute $s_1$. The other eigenfunctions $s_j(x)$ are generally harder to compute using Laplace averages. The eigenfunction $s_1$ can also be estimated with linear algebraic methods (cf.~\cite{mauroy2014global}), or obtained directly from data by using dynamic mode decomposition methods (cf.~\cite{Schmid2010,Tu2014}).

The eigenfunction $s_1(x)$ captures the dominant (i.e. asymptotic) behavior of the unforced system. Hence the boundaries $\partial \cB^\alpha$ of the sets $\cB^\alpha =  \{x \,|\, |s_1(x)| \le \alpha \}$, which are called \emph{isostables}, are important for understanding the dynamics of the system. It can be shown that trajectories with initial conditions on the same isostable $\partial \cB^{\alpha_1}$ converge synchronously toward the fixed point, and reach other isostables $\partial \cB^{\alpha_2}$, with $\alpha_2<\alpha_1$, after a time
\begin{equation}
\label{time_isostable}
\cT = \frac{1}{|\Re(\lambda_1)|} \ln \left(\frac{\alpha_1}{\alpha_2}\right)\,.
\end{equation}
In particular, for $\lambda_1 \in \R$, it follows directly from~\eqref{eq:expansion} that the trajectories starting from $\partial \cB^\alpha$ share the same asymptotic evolution
\begin{equation*}
\phi(t,x,0) \rightarrow x^\ast + v_1 \, \alpha  e^{\lambda_1 t}\,, \quad t\rightarrow \infty\,.
\end{equation*}
Note that isostables could also be defined when the system is driven by an input $u\neq0$, but they are here considered only to describe the dynamics of the unforced system. A more rigorous definition of isostables and more details can be found in~(\cite{mauroy2013isostables}). 

In the case of bistable systems characterized by two equilibria $x^\ast$ and $x^\bullet$ with basins of attraction $\cB(x^\ast)$ and $\cB(x^\bullet)$, respectively, the Koopman operator admits two sets of eigenfunctions $s_k^\ast$ and $s_k^\bullet$. The eigenfunctions $s_k^\ast$ (resp. $s_k^\bullet$) are related to the asymptotic convergence toward $x^\ast$ (resp. $x^\bullet$). 
The dominant eigenfunctions $s_1^\ast$ and $s_1^\bullet$ define two families of isostables, each of which is associated with one equilibrium and lies in the corresponding basin of attraction.

\emph{Monotone Systems and Their Spectral Properties.}
We will study the properties of the system \eqref{sys:f} with respect to a partial order induced by positive cones  in $\R^n$. A set $\cK$ is a \emph{positive cone} if $\Rnn \cK \subseteq \cK$, $\cK + \cK \subseteq \cK$, $\cK\cap -\cK \subseteq \{ 0\}$. A relation $\sim$ is called a {\it partial order} if it is reflexive ($x\sim x$), transitive ($x\sim y$, $y\sim z$ implies $x\sim z$), and antisymmetric ($x\sim y$, $y\sim x$ implies $x = y$). We define a partial order $\succeq_\cK$ through a cone $\cK\in\R^n$ as follows: $x\succeq_\cK y$ if and only if $ x - y \in \cK$. We write $x\not \succeq_\cK y$, if the relation  $x \succeq_\cK y$ does not hold. We also write $x\succ_\cK y$ if $x\succeq_\cK y$ and $x\ne y$, and $x\gg_\cK y$ if $x- y \in \inter(\cK)$. Similarly we can define a partial order on the space of signals $u\in \cU_{\infty}$: $u\succeq_\cK v$ if $u(t) - v(t) \in \cK$ for all $t\ge 0$. 

Systems whose flows preserve a partial order relation $\succeq_\cK$ are called \emph{monotone systems}. We have the following definition.

\begin{defn}\label{def:mon}
The system $\dot x = f(x,u)$ is called \emph{monotone} with respect to the cones $\cK_x$, $\cK_u$ if $\phi(t,x, u)\preceq_{\cK_x} \phi(t,y, v)$ for all $t\ge 0$, and for all $x\preceq_{\cK_x} y$, $u\preceq_{\cK_u} v$. 
\end{defn}
The properties of monotone systems require additional definitions. A function $g:\R^n \rightarrow \R$ is called \emph{increasing} with respect to the cone $\cK$ if $g(x) \le g(y)$ for all $x\preceq_\cK y$. Let $[x,~y]_\cK$ denote the order-interval defined as $[x,~y]_\cK = \{ z | x\preceq_\cK z \preceq_\cK y \}$. A set $\cA$ is called \emph{order-convex} if, for all $x$, $y$ in $\cA$, the interval $[x,~y]_\cK$ is a subset of $\cA$. A set $M$ is called \emph{p-convex} if, for every $x$, $y$ in $M$ such that $x\succeq_{\cK} y$ and every $\lambda\in(0,1)$, we have that $\lambda x+ (1-\lambda) y\in M$. Clearly, order-convexity implies p-convexity. If $\cK = \Rnn^n$ we say that the corresponding partial order is standard. Without loss of generality, we will only consider the standard partial order throughout the paper.
\begin{prop}[\cite{angeli2003monotone}]\label{prop:kamke}
Consider the control system~\eqref{sys:f}, where the sets $\cD$, $\cU$ are p-convex and $f\in C^1(\cD\times \cU)$. Then the system~\eqref{sys:f} is monotone on $\cD\times \cU_{\infty}$ with respect to $\Rnn^n$, $\Rnn^m$ if and only if 
\begin{gather*}
\frac{\partial f_i}{\partial x_j}\ge 0,\quad\forall~i\ne j,\quad(x,u)\in\cl(\cD)\times\cU\\
\frac{\partial f_i}{\partial u_j}\ge 0,\quad\forall~i, j,\quad(x,u)\in\cD\times\cU.
\end{gather*}
\end{prop}

A generalization of this result to other cones can be found in~\cite{angeli2003monotone}. We finally consider the spectral properties of monotone systems that are summarized in the following result. The proof can be found in~\cite{sootla2016basins}.

\begin{prop} \label{prop:mon-eig-fun}
Consider the system $\dot x = f(x)$ with $f\in C^2$, which admits a stable hyperbolic fixed point $x^\ast$ with a basin of attraction $\cB$. Assume that  $\Re(\lambda_1) > \Re(\lambda_j)$ for all $j\ge 2$.  Let $v_1$ be a right eigenvector of the Jacobian matrix $J(x^\ast)$ and let $s_1$ be an eigenfunction corresponding to $\lambda_1$ (with $v_1^T \nabla s_1(x^\ast) = 1$). If the system is monotone with respect to $\Rnn^n$ on $\mathrm{int}(\cB)$, then $\lambda_1$ is real and negative. Moreover, there exist $s_1(\cdot)$ and $v_1$ such that $s_1(x) \ge s_1(y)$ for all $x$, $y\in \cB$ satisfying $x\succeq y$, and $v_1\succ 0 $. 
\end{prop}
This result shows that the sets $\cB_\alpha = \{x| |s_1(x)|\le \alpha \}$ are order-convex for any $\alpha>0$ (cf.~\cite{sootla2016basins}).
\section{Shaping Pulses to Switch Between Fixed Points}
\label{s:shaping-pulses}
In this paper we consider the problem of switching between two stable fixed points by using temporal pulses~\eqref{eq:pulse}. We formalize this problem by making the following assumptions:
  \begin{enumerate}
  \item[{\bf A1.}] Let $f(x,u)$ in~\eqref{sys:f} be continuous in $(x,u)$ and $C^2$ in $x$ for every fixed $u$ on $\cD_f\times \cU$. 
  \item[{\bf A2.}] Let the unforced system~\eqref{sys:f} have two stable hyperbolic fixed points in $\cD_f$, denoted by $x^\ast$ and $x^\bullet$, and let $\cD_f=\cl(\cB(x^\ast)\cup \cB(x^\bullet))$.
  \item[{\bf A3.}] For any $u = \mu h(\cdot, \tau)$ with finite $\mu$ and $\tau$ let $\phi(t, x^\ast, u)$ belong to $\inter(\cD_f)$. Moreover, let there exist $\mu>0$, $\tau>0$ such that $\lim\limits_{t\rightarrow\infty}\phi(t, x^\ast,\mu h(\cdot,\tau))  = x^\bullet$. 
\end{enumerate}

Assumption~{A1} guarantees existence and uniqueness of solutions, while Assumption~A2 defines a bistable system. Note that in~\cite{sootla2015pulsesaut}, the assumptions A1--A2 are less restrictive. That is, $f(x,u)$ is Lipschitz continuous in $x$ for every fixed $u$, and the fixed points are asymptotically stable. Our assumptions are guided by the use of the Koopman operator.  Assumptions~A1 and~A2 guarantee the existence of eigenfunctions $s_1^\ast(x)$ and $s_1^\bullet(x)$ that are continuously differentiable on each basin of attraction. Assumption~{A3} is technical and ensures that the switching problem is feasible.

The goal of our \emph{control problem} is to characterize the so-called switching set $\cS$ defined as
\begin{gather} \label{switching-set}
\cS = \left\{(\mu, \tau) \in \Rp^2 \,\Bigl| \forall t > \tau:~\phi(t, x^\ast,\mu h(\cdot,\tau)) \in \cB(x^\bullet) \right\}.
\end{gather} 
It is shown in~\cite{sootla2015pulsesaut} that the set $\cS$ is simply connected and order-convex under some assumptions, a property which is useful to obtain an efficient computational procedure. In particular, the boundary $\partial \cS$, called the switching separatrix, is such that for all $(\mu,\tau)$, $(\nu, \xi)$ in $\partial \cS$ we cannot have that $\mu > \nu$ and $\tau > \xi$. The following result sums up one of the theoretical contribution in~\cite{sootla2015pulsesaut}. 
\begin{prop}\label{prop:comp-sys-sw}
Let the system~$\dot x = f(x,u)$ satisfy Assumptions~A1--A3. The following conditions are equivalent:

(i) the set $\cS$ is order-convex and simply connected;

(ii) let $\phi(\tau_1, x^\ast, \mu_1 h(\cdot, \tau_1))$ belong to $\cB(x^\bullet)$, then the flow $\phi(\tau_2, x^\ast, \mu_2 h(\cdot, \tau_2))\in\cB(x^\bullet)$ for all $\mu_2 \ge \mu_1$, $\tau_2 \ge \tau_1$.

Moreover, if the system~$\dot x = f(x,u)$ is monotone with respect to $\cK\times\R$ on $\cD\times\cU_\infty$ and satisfies Assumptions~A1--A3, then (i)-(ii) hold.
\end{prop}

We also note that the results in~\cite{sootla2015pulsesaut} were extended to account for parametric uncertainty in the vector field under additional constraints. In particular, it is possible to estimate bounds on the switching set $\cS$.

Now, we proceed by providing an operator-theoretic point of view on shaping pulses, which allows to study rates of convergence to the fixed point.
\begin{defn}
Let $\cS\subseteq \Rp^2$ be a set of $(\mu, \tau)$ such that $\phi(\tau, x^\ast, \mu)\in \cB(x^\bullet)$. We define the \emph{switching function} $r: \cS \mapsto \C$ by
\begin{equation*}
r(\mu,\tau) = s_1^\bullet(\phi(\tau, x^\ast, \mu ))
\end{equation*}
for all $(\mu,\tau)$ such that $\phi(\tau, x^\ast, \mu)\in \cB(x^\bullet)$.
\end{defn}
The level sets $\partial \cS^\alpha$ of $|r|$ defined as
\begin{gather*}
\partial \cS^\alpha = \left\{ (\mu,\tau)  \in \cS \,\Bigl| |r(\mu,\tau)| = \alpha \right\}, \quad \alpha \geq 0
\end{gather*}
are reminiscent of the isostables $\partial \cB^\alpha$, which are the level sets of $|s_1|$. We also consider the sublevel sets of $|r|$
\begin{gather*}
\cS^\alpha = \left\{ (\mu,\tau) \in \Rp^2 \,\Bigl| |r(\mu,\tau)|  \le \alpha \right\}
\end{gather*}
and it is straightforward to show that the switching set $\cS$ in~\eqref{switching-set} is equal to $\cS^\infty=\bigcup_{\alpha\ge 0} \cS^\alpha$. The level sets $\partial \cS^\alpha$ can therefore be seen as a generalization of the switching separatrix $\partial \cS=\partial \cS^\infty$.

The level sets $\partial \cS^\alpha$ capture the pairs $(\mu,\tau)$ such that the trajectories $\phi(t, x^\ast, \mu h(\cdot,\tau))$ reach the isostable $\partial \cB^\alpha$ at time $t=\tau$ and cross the same isostables for all $t\ge \tau$. This implies that the trajectories with the pairs $(\mu,\tau)$ on the same level sets $\partial \cS^\alpha$ will take the same time to converge towards the fixed point $x^\bullet$ when the control is switched off. We can for instance estimate the time $\cT$ needed to reach the set $\cB^\varepsilon$ for a positive $\varepsilon$. It follows from~\eqref{time_isostable} that, for $(\mu,\tau)\in \partial \cS^\alpha$, we have
\begin{gather}
\cT(\mu,\tau,\varepsilon) = \frac{1}{|\lambda_1|} \ln \left(\frac{\alpha}{\varepsilon} \right),
\end{gather}
where a negative $\cT$ means that the trajectory is inside the set $\cB^\varepsilon$ at time $t=\tau$. Hence, the quantity $\cT_{\rm tot} = \cT + \tau$ is the time it takes to reach $\cB^\varepsilon$ if $\cT$ is nonnegative. For small enough $\varepsilon$, the function $\cT_{\rm tot}(\mu,\tau,\varepsilon)$ approximates the amount of time required to reach a small neighborhood of the fixed point $x^\bullet$. 

In order to compute the switching function $r$, we can again employ Laplace averages
\begin{multline}
r(\mu,\tau)  = \lim\limits_{T\rightarrow\infty} \frac{1}{T} \int\limits_0^T g\circ \phi(t, \phi(\tau,x^\ast, \mu),0) e^{-\lambda_1 t}d t \\
 = \lim\limits_{T\rightarrow\infty} \frac{1}{T}\int\limits_\tau^T g\circ \phi(t, x^\ast, \mu h(\cdot,\tau)) e^{-\lambda_1 (t-\tau)}d t ,
 \label{eigen-fun-cont}
\end{multline}
where $\lambda_1$ is the dominant Koopman eigenvalue and $g$ satisfies $g(x^\bullet)=0$ and $ v_1^T \nabla g(x^\bullet)\neq 0$. Note again that the limit does not converge unless $\phi(\tau,x^\ast, \mu)$ belongs to $\cB(x^\bullet)$.

The computation of $\cS^\alpha$ is not an easy task in general, but certainly possible. However, additional assumptions on the system simplify the computation of these sets. From this point on we will assume that $r(\mu,\tau)$ has only real values (i.e. $s_1^\bullet \in \mathbb{R}$), which holds if the dominant Koopman eigenvalue on $\cB(x^\bullet)$ is real (see~\cite{mauroy2013isostables}). In this case, the set $\partial \cS^\alpha$ can be split into two sets 
\begin{align*}
\partial_- \cS^\alpha &= \left\{ (\mu,\tau)  \in \Rp^2 \,\Bigl| r(\mu,\tau) = -\alpha \right\}, \\
\partial_+ \cS^\alpha &= \left\{ (\mu,\tau)  \in \Rp^2 \,\Bigl| r(\mu,\tau) = \alpha \right\}.
\end{align*}
If $\cS^\alpha$ is order-convex (as it is shown below for the case of monotone systems), then  $\partial_- \cS^\alpha$ and  $\partial_+ \cS^\alpha$ are the sets of minimal and maximal elements of $\cS^\alpha$, respectively. That is, if $x \ll y$ for some $x \in\partial_+ \cS^\alpha$ (respectively, if $x\gg y$ for some $x\in\partial_-^\alpha \cS$) then $y\not\in\cS$. This implies that $\partial_- \cS^\alpha$ and $\partial_+ \cS^\alpha$ are monotone maps, which significantly facilitates computations of $\cS^\alpha$ by applying the algorithm from~\cite{sootla2016basins} with a minor modification.


Monotonicity also plays a role in the properties of the sublevel sets $\cS^\alpha$, as it does in the properties of the switching separatrix. The main result of the section establishes that, for monotone systems, the sets $\cS^\alpha$ are order-convex and $\partial \cS^\alpha$ are monotone maps.

\begin{thm} \label{thm:mon-switch-iso-mon-sys}
Let the system~\eqref{sys:f} satisfy Assumptions~A1--A3 and be monotone on $\cD\times \cU^\infty$. Then

(i) the set $\cS^\alpha$ is order-convex (with respect to the positive orthant) for any non-zero $\alpha$;

(ii) the set $\partial_+ \cS^\alpha$ is a monotone map, that is for all $(\mu_1, \tau_1)$, $(\mu_2, \tau_2)\in \partial_+ \cS^\alpha$, if $\tau_1 < \tau_2$ then $\mu_1 \ge \mu_2$, and if $\mu_1 < \mu_2$ then $\tau_1 \ge \tau_2$. Moreover, the set $\partial_- \cS^\alpha$ is a graph of a monotonically decreasing function for any finite non-zero $\alpha$;

(iii) if additionally $\phi(t , x, \mu) \gg \phi(t , x, \nu)$ for all $x$, all $\mu > \nu\ge 0$ and all $t>0$, then $\partial_- \cS^\alpha$ and $\partial_+ \cS^\alpha$ are graphs of monotonically decreasing functions for any finite non-zero $\alpha$.
\end{thm}
The proof of Theorem~\ref{thm:mon-switch-iso-mon-sys} is in Appendix~\ref{app:proof}. An interesting detail is that the level sets $\partial_-\cS^\alpha$ are graphs of decreasing functions. This implies that the switching separatrix $\partial\cS^\infty$ can be approximated by a graph of a function by setting $\alpha \gg 0$. We can also partially recover the results in~\cite{sootla2015pulsesaut} by letting $\alpha\rightarrow +\infty$. Note, however, that $\partial_- \cS^{\infty}$ is not necessarily a graph of a function, since strict inequalities may no longer hold in the limit. 
\section{Examples}
\label{s:example}

\emph{Eight Species Generalized Repressilator.} This system is an academic example (cf.~\cite{Strelkowa10}), where each of the species represses another species in a ring topology. The corresponding dynamic equations for a symmetric generalized repressilator are as follows: 
  \begin{equation}
    \label{eq:gr}
    \begin{aligned}
   \dot x_1 &= \frac{p_{1}}{1 + (x_{8}/p_2)^{p_3}} + p_4 - p_5 x_1 + u, \\
   \dot x_2 &= \frac{p_{1}}{1 + (x_{1}/p_2)^{p_3}} + p_4 - p_5 x_2, \\
   \dot x_i &= \frac{p_{1}}{1 + (x_{i-1}/p_2)^{p_3}} + p_4 - p_5 x_i,~\forall i = 3,\dots 8,  
    \end{aligned}
  \end{equation}
where $p_1 = 100$, $p_2 = 1$, $p_3 = 2$, $p_4 = 1$, and $p_5 = 1$. This system has two stable equilibria $x^\ast$ and $x^\bullet$ and is monotone with respect to the cones $\cK_x = P_x \R^8$ and $\cK_u = \R$, where $P_x= \diag([1,~-1,~1,~-1,~1,~-1,~1,~-1])$. We have also $x^\ast \preceq_{\cK_x} x^\bullet$. It can be shown that the unforced system is strongly monotone in the interior of $\Rnn^8$ for all positive parameter values. One can also verify that there exist pulse control signals $u$ that switch the system from the state $x^\ast$ to the state $x^\bullet$.

The level sets $\partial \cS^\alpha$ are depicted in Figure~\ref{fig:switching-iso}, where instead of the values of the level sets we provide the time needed to converge to $\partial\cB^{0.01}$. 
As the reader may notice, the two curves related to $\cT=5$ (i.e. blue solid curves) lie close to each other. They approximate the pairs $(\mu,\tau)$ that drive the flow to the zero level set of $s_1^\bullet(x)$. It is also noticeable that the level sets $\partial \cS^\alpha$ are less dense on the right of these lines. This is explained by the fact that the flow is driven by the pulse beyond the zero level set of $s_1^\bullet(x)$ and has to counteract the dynamics of the system.

The generalized repressilator is a monotone system, and hence the premise of Theorem~\ref{thm:mon-switch-iso-mon-sys} is fulfilled. The level sets $\partial \cS^\alpha$ in Figure~\ref{fig:switching-iso} appear to be graphs of monotonically decreasing functions, an observation which is consistent with the claim of Theorem~\ref{thm:mon-switch-iso-mon-sys}.

%

\emph{Toxin-antitoxin system.} Consider the toxin-antitoxin system studied in~\cite{cataudella2013conditional}. 
\begin{align*}
\dot T &= \frac{\sigma_T}{\left(1 + \frac{[A_f][T_f]}{K_0}\right)(1+\beta_M [T_f])} - \dfrac{1}{(1+\beta_C [T_f])} T \\
\dot A &= \frac{\sigma_A}{\left(1 + \frac{[A_f][T_f]}{K_0}\right)(1+\beta_M [T_f])} - \Gamma_A  A + u \\
\varepsilon [\dot A_f] &= A - \left([A_f] + \dfrac{[A_f] [T_f]}{K_T} + \dfrac{[A_f] [T_f]^2}{K_T K_{T T}}\right) \\ 
\varepsilon [\dot T_f] &= T - \left([T_f] + \dfrac{[A_f] [T_f]}{K_T} + 2 \dfrac{[A_f] [T_f]^2}{K_T K_{T T}}\right), 
\end{align*}
where $A$ and $T$ is the total number of toxin and antitoxin proteins, respectively, while $[A_f]$, $[T_f]$ is the number of free toxin and antitoxin proteins.
In~\cite{cataudella2013conditional}, the authors considered the model with $\varepsilon = 0$. In order to simplify our analysis we set $\varepsilon = 10^{-6}$.
For the parameters
\begin{gather*}
\sigma_T = 166.28,~~K_0 = 1,~~\beta_M = \beta_c =0.16,~~\sigma_A = 10^2\\
\Gamma_A = 0.2,~~K_T = K_{TT} = 0.3,
\end{gather*}
the system is bistable with two stable steady states:
\begin{gather*}
 x^{\bullet} = \begin{pmatrix}  27.1517 &  80.5151 & 58.4429 & 0.0877  \end{pmatrix} \\ 
 x^\ast  = \begin{pmatrix} 162.8103 & 26.2221  &  0.0002 & 110.4375 \end{pmatrix}. 
\end{gather*}

It can be verified that the system is not monotone with respect to any orthant, however, it was established in~\cite{sootla2015koopman} that it is \emph{eventually monotone}. This means that the flow satisfies the monotonicity property after some initial transient.
\begin{figure}[t]
\centering
  \includegraphics[width = 0.9\columnwidth]{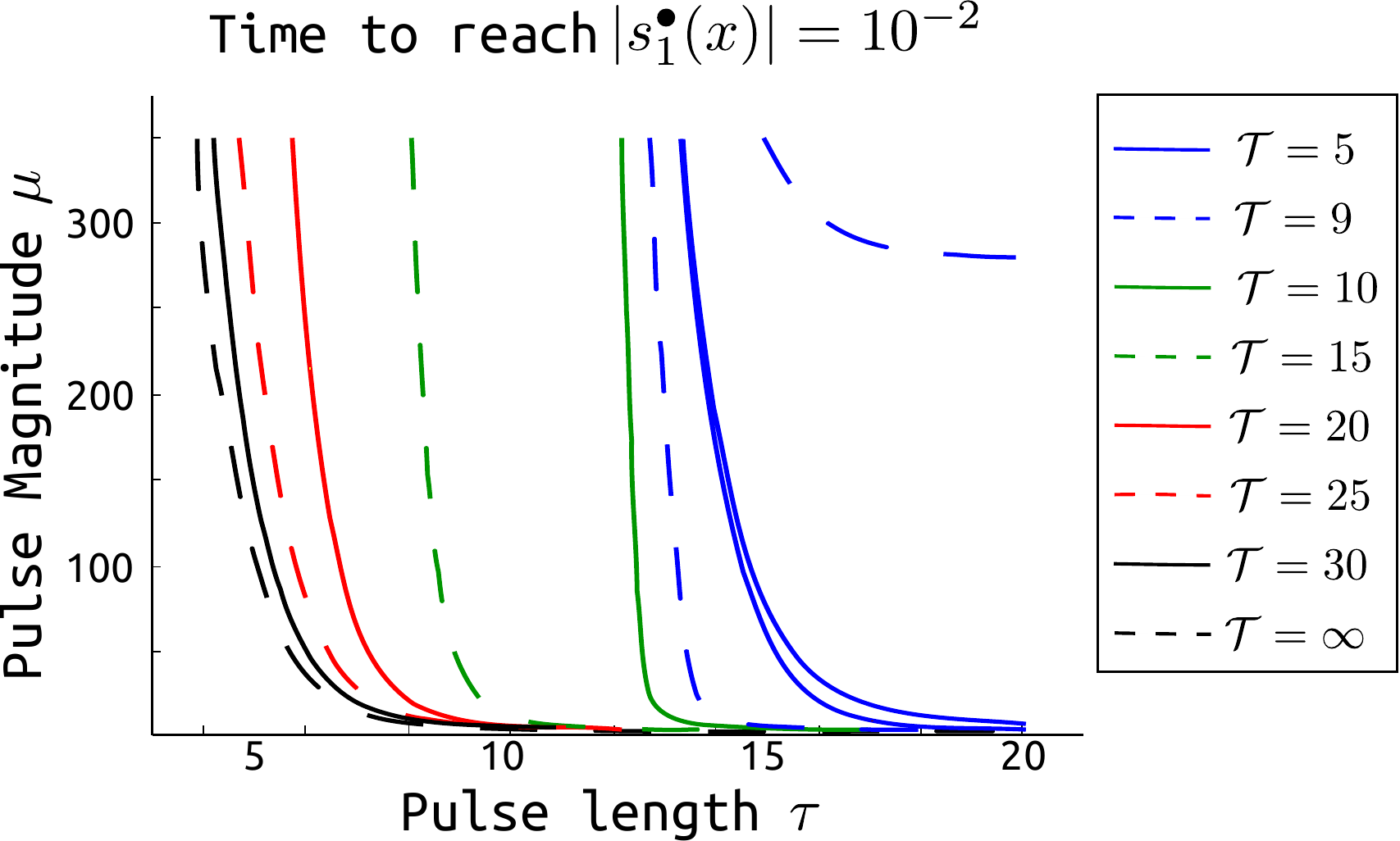}
\caption{The level sets of $\cT = \frac{1}{|\lambda_1|} \ln \left(\frac{|r(\mu,\tau)|}{\varepsilon} \right)$, where $\varepsilon = 10^{-2}$. The pairs $(\mu,\tau)$ on the same curve are related to trajectories which converge synchronously to the stable equilibrium.} \label{fig:switching-iso}
\end{figure}
\begin{figure}[t]\centering
  \includegraphics[width = 0.9\columnwidth]{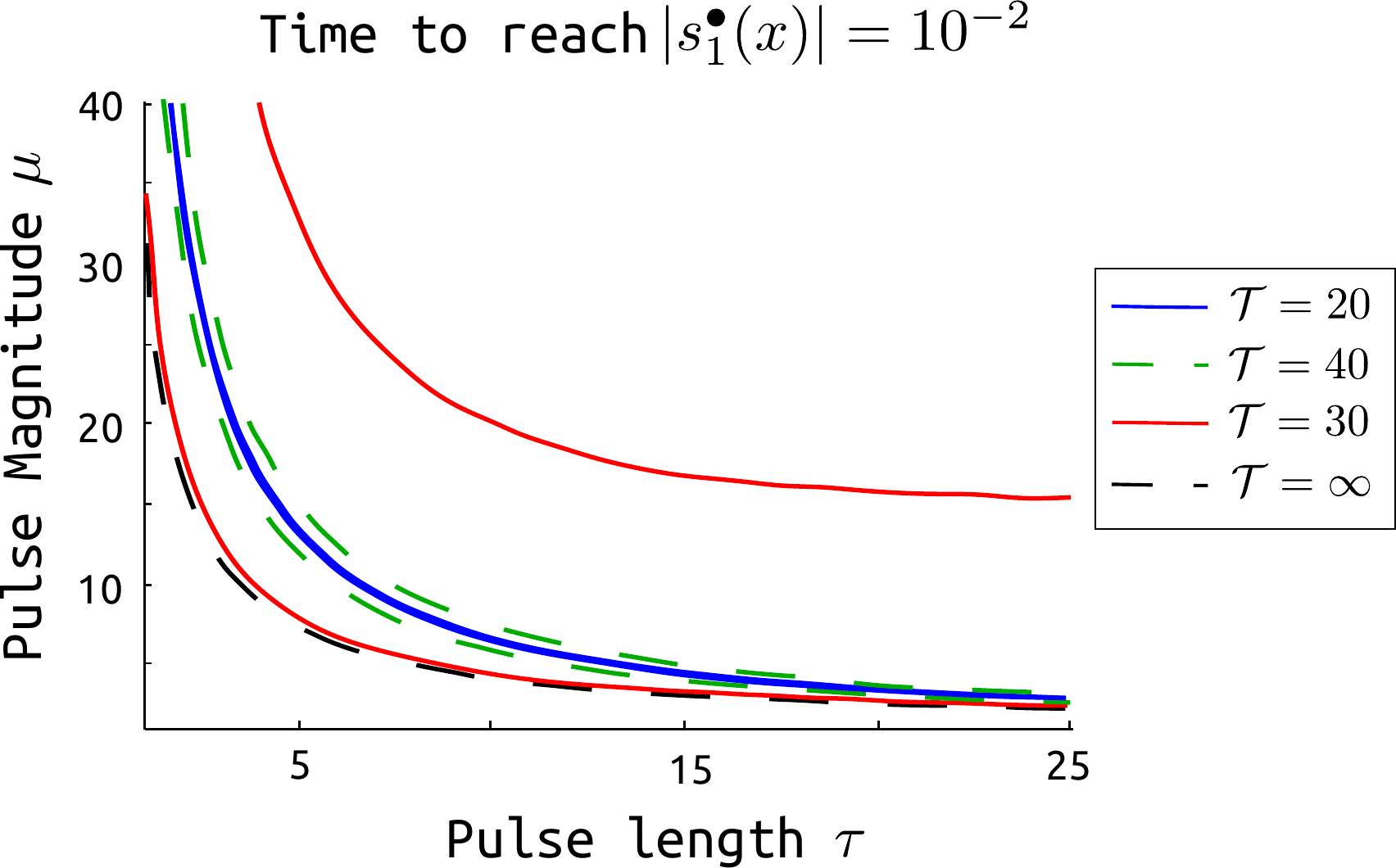}
\caption{The level sets of $\cT = \frac{1}{|\lambda_1|} \ln \left(\frac{|r(\mu,\tau)|}{\varepsilon} \right)$ for the toxin-antitoxin system, where $\varepsilon = 10^{-2}$.}\label{fig:tat-ss}
\end{figure}

We depict the level sets $\partial \cS^\alpha$ in Figure~\ref{fig:tat-ss}, where it appears that these sets are monotone curves although the system does not satisfy the assumptions of Theorem~\ref{thm:mon-switch-iso-mon-sys}. This could be explained by the property of \emph{eventual monotonicity}, but we have not further investigated this case. 

\emph{Lorenz System.}
Now we illustrate the level sets $\partial \cS^\alpha$ in the case where $r$ is complex-valued. Consider the Lorenz system 
\begin{align*}
\dot x_1 &= \sigma (x_2 - x_1)  + u\\
\dot x_2 &= x_1 (\rho - x_3) - x_2 + u\\
\dot x_3 &= x_1 x_2 -\beta x_3
\end{align*}
with parameters $\sigma = 10$, $\rho = 2$, $\beta = 8/3$, which is bistable but not monotone. Note that the Jacobian matrix at the steady states has two complex conjugate dominant eigenvalues. In this case, Theorem~\ref{thm:mon-switch-iso-mon-sys} cannot be applied and Figure~\ref{fig:lorenz-ss} shows that the level sets $\partial \cS^\alpha$ are not monotone. It is however noticeable that the lower part of the switching separatrix (black curve) seems to be monotone (but the upper part is not monotone). 

\begin{figure}[t]\centering
  \includegraphics[width = 0.9\columnwidth]{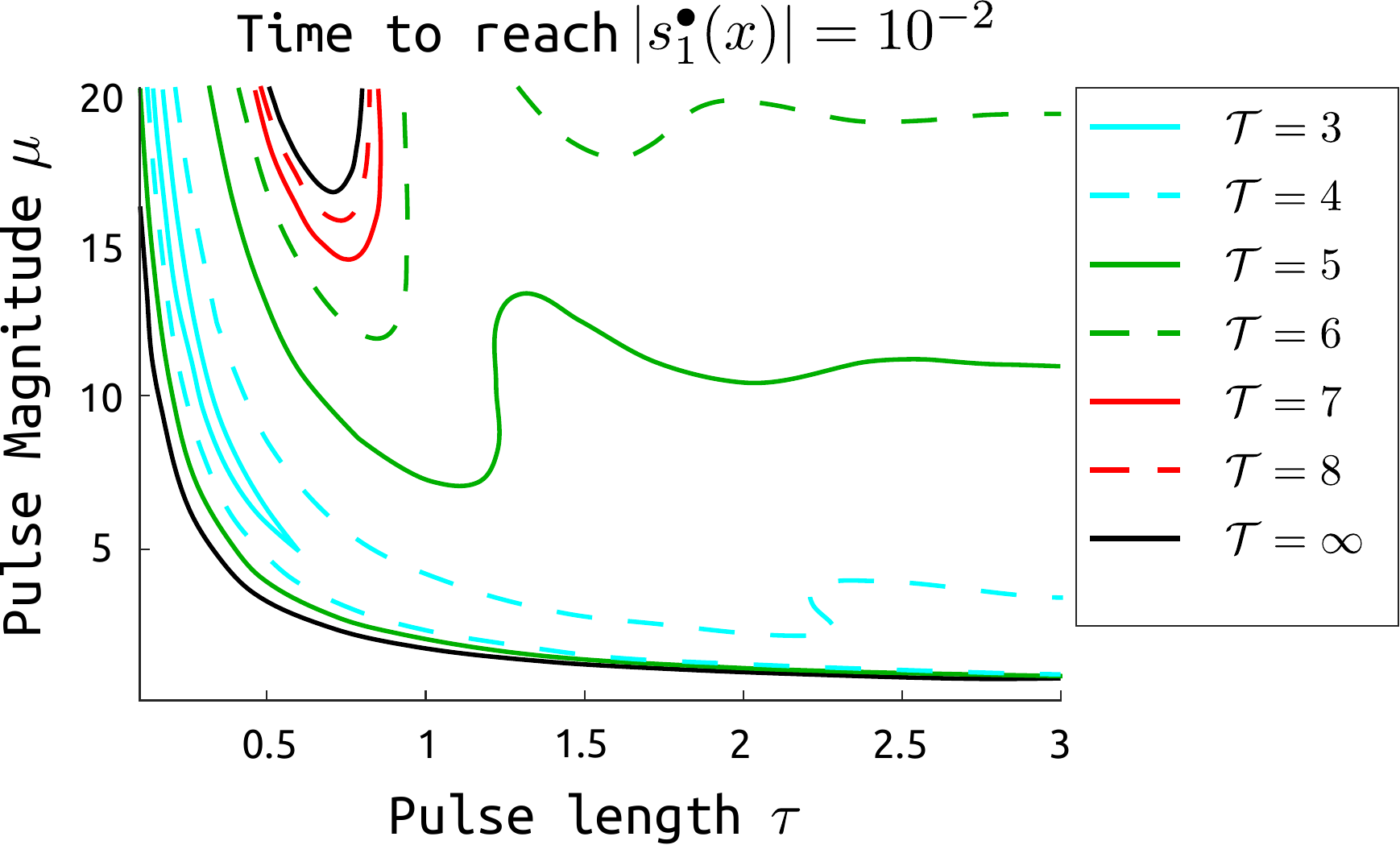}
\caption{The level sets of $\cT = \frac{1}{|\lambda_1|} \ln \left(\frac{|r(\mu,\tau)|}{\varepsilon} \right)$ for the Lorenz system, where $\varepsilon = 10^{-2}$. } \label{fig:lorenz-ss}
\end{figure}

\section{Conclusion}

In this paper, we have further developed a recent study on the problem of switching a bistable system between its steady states with temporal pulses. We have introduced a family of curves in the control parameter space, denoted as $\partial \cS^\alpha$, which provide an information on the time needed by the system to converge to the steady state. The sets $\partial \cS^\alpha$ can be viewed as an extension of the switching separatrix defined in the previous study. They are related to the dominant eigenfunction of the Koopman operator, a property that provides a method to compute them. In the case of monotone systems, we have also shown that the level sets $\partial \cS^\alpha$ are characterized by strong (monotonicity) properties.

Future research will investigate the topological properties of the level sets $\partial \cS^\alpha$ such as connectedness. Moreover, characterizing the properties of the level sets $\partial \cS^\alpha$ (and the switching separatrix) in the case of non-monotone (e.g. eventually monotone) systems is still an open question.

\bibliography{bibl_koopman}
\appendix
\section{Proof of Theorem~\ref{thm:mon-switch-iso-mon-sys}} \label{app:proof}
Before we proceed with the proof of Theorem~\ref{thm:mon-switch-iso-mon-sys}, we show a technical result, which establishes that for monotone systems the transient during switching between operating points is always an increasing function.
\begin{prop}\label{prop:pos-der}
Let the system~\eqref{sys:f} be monotone on $\cD\times \cU_\infty$, then
\begin{gather} \label{cond-pos-der}
\phi(\xi + h, x^\ast, \mu) \succeq \phi(\xi, x^\ast, \mu)
\end{gather}
for any nonnegative scalars $h$, $\xi$, $\mu$.
\end{prop}

\begin{proof} The proof stems from a well-know result in monotone systems theory, which states that the flow cannot increase (or decrease) on two disjoint time intervals. We show this result for completeness. Due to monotonicity we have 
\begin{gather}
\phi(h, x^\ast, \mu) \succeq \phi(h, x^\ast, 0) = x^\ast
\end{gather} 
for all nonnegative $h$, $\mu$. Therefore, the semigroup property of the dynamical systems implies
\begin{multline}
\phi(\xi+h, x^\ast, \mu) = \phi(\xi, \phi(h, x^\ast, \mu),\mu) \succeq \phi(\xi,x^\ast,\mu),
\end{multline}
for any nonnegative scalars $\xi$, $h$, $\mu$. 
\end{proof}

\begin{proof-of}{Theorem~\ref{thm:mon-switch-iso-mon-sys}}
(i) Let $(\mu_1, \tau_1)$, $(\mu_2, \tau_2)$ belong to $\cS^\alpha$ for some finite $\alpha>0$ and $\mu_2 \ge \mu_1$, $\tau_2 \ge \tau_1$. Then $r(\mu_1, \tau_1)$, $r(\mu_2, \tau_2)$ are finite. Due to monotonicity and Proposition~\ref{prop:pos-der}, we have that 
\begin{multline*}
\phi(\tau_1, x^\ast, \mu_1 h(\cdot, \tau_1)) \preceq 
\phi(\tau_1, x^\ast, \mu_2 h(\cdot, \tau_2)) \preceq \\
\phi(\tau_2, x^\ast, \mu_2 h(\cdot, \tau_2)),
\end{multline*}
which according to Proposition~\ref{prop:mon-eig-fun} implies that 
\begin{gather*}
s_1^\bullet(\phi(\tau_1, x^\ast, \mu_1 h(\cdot, \tau_1))) \le s^\bullet_1(\phi(\tau_2, x^\ast, \mu_2 h(\cdot, \tau_2))).
\end{gather*}
Using this property it is rather straightforward to show that $\cS^\alpha$ is order-convex.

(ii) The sets $\partial_+ \cS^\alpha$ and $\partial_- \cS^\alpha$ contain the maximal and minimal elements, respectively, of the order-convex set $\cS^\alpha$. Assume that $(\mu_1,\tau_1),(\mu_2,\tau_2) \in \partial_+ \cS^\alpha$ (or $(\mu_1,\tau_1),(\mu_2,\tau_2) \in \partial_- \cS^\alpha$). Then we cannot have $(\mu_1,\tau_1) \ll (\mu_2,\tau_2)$. Hence, if $\tau_1 < \tau_2$, we must have $\mu_1 \geq \mu_2$ and if $\mu_1 < \mu_2$, we must have $\tau_1 \geq \tau_2$. We prove the second part of the statement by contradiction. Let $\tau_1 < \tau_2$, $\mu_1 \le \mu_2$  and let $s_1^\bullet(\phi(\tau_1, x^\ast,\mu_1 h(\cdot,\tau_1))) = -\alpha$ and $s^\bullet_1(\phi(\tau_2, x^\ast,\mu_2 h(\cdot,\tau_2))) = -\alpha$, where $\alpha> 0$. Due to monotonicity we have that 
\[
\phi(t, x^\ast,\mu_1 h(\cdot,\tau_1))  \preceq \phi(t, x^\ast,\mu_2 h(\cdot,\tau_2)) \,\, \forall t \ge 0,
\]
which according to Proposition~\ref{prop:mon-eig-fun} entails
\begin{gather}
\label{cond:stab}
	s_1^\bullet(\phi(t, x^\ast,\mu_1 h(\cdot,\tau_1))) \le s^\bullet_1(\phi(t, x^\ast,\mu_2 h(\cdot,\tau_2))).
\end{gather}

The flow $\phi(t, x^\ast,\mu_1 h(\cdot,\tau_1))$ converges to $x^\bullet$ freely for all $t> \tau_1$, since $h(t,\tau_1) = 0$ for all $t> \tau_1$. Negativity of $-\alpha$ implies that $s_1^\bullet(\cdot)$ is growing along the trajectory $\phi(\tau_1, x^\ast,\mu_1 h(\cdot,\tau_1))$ and 
hence $s_1^\bullet(\phi(\tau_2, x^\ast, \mu_1 h(\cdot,\tau_1))) > -\alpha$. This, however, contradicts~\eqref{cond:stab}, since 
\[
s_1^\bullet(\phi(\tau_2, x^\ast,\mu_1 h(\cdot,\tau_1))) \le s_1^\bullet(\phi(\tau_2, x^\ast,\mu_2 h(\cdot,\tau_2))) = -\alpha.
\]

(iii) Let  $r(\mu_1, \tau_1)= r(\mu_2, \tau_2) = \alpha$ and pick  $\tau_1 < \tau_2$. Assume that $\mu_1 \le \mu_2$. We have that
\begin{multline*}
\phi(\tau_1, x^\ast,\mu_1 h(\cdot,\tau_1)) \preceq \phi(\tau_1, x^\ast,\mu_2 h(\cdot,\tau_1)) \preceq \\
\phi(\tau_2, x^\ast,\mu_2 h(\cdot,\tau_1))
\end{multline*}
where the first inequality follows from monotonicity and the second follows from Proposition~\ref{prop:pos-der}. Due to the condition $\phi(t , x, \mu) \gg \phi(t , x, \nu)$ for all $x$, all $\mu > \nu$ and all $t>0$ in the premise, we have
\begin{align*}
\phi(\tau_1, x^\ast,\mu_1 h(\cdot,\tau_1)) &\preceq \phi(\tau_2, x^\ast,\mu_2 h(\cdot,\tau_1)) \\
&=\phi(\tau_2 - \tau_1, \phi(\tau_1, x^\ast, \mu_2 h(\cdot,\tau_1)), 0)  \\
&\ll \phi(\tau_2 - \tau_1, \phi(\tau_1, x^\ast, \mu_2 h(\cdot,\tau_1)), \mu_2 ) \\
&= \phi(\tau_2, x^\ast,\mu_2 h(\cdot,\tau_2))
\end{align*}
or equivalently $r(\mu_1,\tau_1) < r(\mu_2,\tau_2)$. We arrive at a contradiction, hence $\mu_1 > \mu_2$, which implies that $\partial_+ \cS^\alpha$ is a graph of a decreasing function.
\end{proof-of}

\end{document}